\documentclass[a4paper]{article}
\usepackage[letterpaper, left=1in, right=1in, bottom=1in, top=1in]{geometry}
\usepackage[T1]{fontenc}
\usepackage{authblk}
\usepackage[dvipdfmx]{graphicx}
\usepackage{amsthm,amsmath,cases}
\usepackage{booktabs}
\usepackage{microtype}
\usepackage{multirow}
\usepackage[dvipdfmx,bookmarksdepth=4,bookmarksnumbered=true]{hyperref}

\newtheorem{theorem}{Theorem}
\newtheorem{corollary}{Corollary}
\newtheorem{lemma}{Lemma}
\newtheorem{observation}{Observation}

\newcommand{\numocc}{\#}
\newcommand{\occ}{occ}
\newcommand{\MUS}{\mathsf{MUS}}
\newcommand{\MMUS}{\mathsf{MMUS}}
\newcommand{\SUS}{\mathsf{SUS}}
\newcommand{\Pred}{\mathsf{Pred}}
\newcommand{\Succ}{\mathsf{Succ}}
\newcommand{\RmQ}{\mathsf{RmQ}}
\newcommand{\RMQ}{\mathsf{RMQ}}

\newcommand{\pcover}{\mathit{cover}}
\newcommand{\icover}{\mathit{cover}}
\newcommand*{\susname}[1]{{{\renewcommand{\rmdefault}{ptm}\fontfamily{ptm}\selectfont\textrm{\textup{#1}}}}} \newcommand{\lmSUS}{\susname{lmSUS}}
\newcommand{\rmSUS}{\susname{rmSUS}}
\newcommand{\lmMUS}{\susname{lmMUS}}
\newcommand{\rmMUS}{\susname{rmMUS}}
\newcommand{\MB}{\mathsf{MB}}
\newcommand{\ME}{\mathsf{ME}}
\newcommand{\SA}{\mathsf{SA}}
\newcommand{\ISA}{\mathsf{ISA}}
\newcommand{\LCP}{\mathsf{LCP}}
\newcommand{\PLCP}{\mathsf{PLCP}}
\newcommand{\succPLCP}{\mathit{succPLCP}}
\newcommand{\RankPrev}{\mathsf{RankPred}}
\newcommand{\RankNext}{\mathsf{RankSucc}}
\newcommand{\length}{\mathsf{L}}
\newcommand{\MUSbegin}{\mathsf{B}}
\newcommand{\lengthdiff}{\mathsf{LD}}
\newcommand{\predonepos}{\mathit{pred}\mathsf{1}\mathit{pos}}
\newcommand{\succonepos}{\mathit{succ}\mathsf{1}\mathit{pos}}
\newcommand{\predneq}{\mathit{predneq}}
\newcommand{\succneq}{\mathit{succneq}}
\newcommand{\X}{\mathsf{X}}
\newcommand{\Y}{\mathsf{Y}}
\newcommand{\MUSlen}{\mathsf{MUSlen}}
\newcommand{\rank}{\mathit{rank}}
\newcommand{\select}{\mathit{select}}
\newcommand{\nil}{\mathit{nil}}

\begin{document}
\title{Space-Efficient Algorithms for Computing Minimal/Shortest Unique Substrings}
\author[1]{Takuya~Mieno}
\author[1,2]{Dominik~K\"oppl}
\author[1]{Yuto~Nakashima}
\author[1,3]{Shunsuke~Inenaga}
\author[4]{Hideo~Bannai}
\author[1]{Masayuki~Takeda}
\affil[1]{Department of Informatics, Kyushu University.

\texttt{\{takuya.mieno,dominik.koeppl,\\yuto.nakashima,inenaga,takeda\}@inf.kyushu-u.ac.jp}}
\affil[2]{Japan Society for the Promotion of Science.}
\affil[3]{PRESTO, Japan Science and Technology Agency,}
\affil[4]{M\&D Data Science Center, Tokyo Medical and Dental University.

\texttt{hdbn.dsc@tmd.ac.jp}}
\date{}
\maketitle
\begin{abstract}
  Given a string $T$ of length~$n$,
  a substring $u = T[i..j]$ of~$T$ is called
  a shortest unique substring (SUS) for an interval $[s,t]$
  if
  (a) $u$ occurs exactly once in $T$,
  (b) $u$ contains the interval $[s,t]$~(i.e. $i \leq s \leq t \leq j$), and
  (c) every substring $v$ of $T$ with $|v| < |u|$ containing $[s,t]$ occurs at least twice in $T$.
  Given a query interval $[s, t] \subset [1, n]$,
  the \emph{interval SUS problem} is to output all the SUSs for the interval $[s,t]$.
  In this article, we propose a $4n + o(n)$ bits data structure answering an interval SUS query
  in output-sensitive $O(\occ)$ time, where $\occ$ is the number of returned SUSs.
  Additionally, we focus on the \emph{point SUS problem}, which is the interval SUS problem for $s = t$.
  Here, we propose a $\lceil (\log_2{3} + 1)n \rceil + o(n)$ bits data structure
  answering a point SUS query in the same output-sensitive time.
  We also propose space-efficient algorithms for computing the \emph{minimal unique substrings} of~$T$.
\end{abstract}

\section{Introduction} \label{sec:intro}
A substring $u = T[i..j]$ of a string $T$ is called
a \emph{shortest unique substring~(SUS)} for an interval $[s,t]$ if
(a) $u$ occurs exactly once in $T$,
(b) $u$ contains the interval $[s,t]$~(i.e., $i \leq s \leq t \leq j$), and
(c) every substring $v$ of $T$ with $|v| < |u|$ containing $[s,t]$ occurs at least twice in $T$.
Given a query interval $[s, t] \subset [1, n]$,
the \emph{interval SUS problem} is to output all the SUSs for $[s,t]$.
When a query interval consists of a single position (i.e., $s=t$),
the SUS problem becomes a so-called \emph{point} SUS problem.

\paragraph*{\bf Point SUS Problem}
The point SUS problem was introduced by Pei et al.~\cite{Pei2013SUS}.
This problem is motivated by applications in bioinformatics like
genome comparisons~\cite{Haubold2005SUS} or PCR primer design~\cite{Pei2013SUS}.
Pei et al. tackled this problem with an $O(n)$ words data structure that can return one SUS for a given query position in constant time.
They can compute this data structure in $O(n^2)$ time with $O(n)$ space.
Based on that result,  Tsuruta et al.~\cite{Tsuruta2014SUS} provided an $O(n)$ words
data structure answering the same query (returning one SUS) in constant time.
Their data structure can be constructed in $O(n)$ time.
{\.I}leri et al.~\cite{Ileri2015SUS} independently showed another data structure
with the same time complexities.
For the general point SUS problem, Tsuruta et al.~\cite{Tsuruta2014SUS}
can also resort to their proposed data structure returning all SUSs for a query position
in optimal $O(\occ)$ time, where $\occ$ is the number of returned SUSs.

The aforementioned data structures all take $\Theta(n)$ words.
This space can become problematic for large $n$.
This problem was perceived by Hon et al.~\cite{Hon2015inplaceSUS},
who proposed a data structure consisting of
the input string $T$ and two integer arrays, each of length $n$.
Both arrays store, respectively, the beginning and the ending position of
a SUS for each position $i$ with $1 \leq i \leq n$.
Hon et al.\ provided an algorithm that can construct these two arrays in linear time
with $O(\log n)$ bits of additional working space,
given that both arrays are stored in $2n \log n$ bits and that $\sigma \le n$.
Instead of building a data structure,
Ganguly et al. \cite{Ganguly2017SUS} proposed a time-space trade-off algorithm
using $O(n/\tau)$ words of additional working space,
answering a given query in $O(n\tau^2\log{\frac{n}{\tau}})$ time directly, for a trade-off parameter~$\tau \ge 1$.
They also proposed the first \emph{compact} data structure of size $4n+o(n)$ bits
that can answer a query in constant time.
They can construct this data structure in $O(n\log n)$ time using $O(n\log\sigma)$ bits of
additional working space.

\paragraph*{\bf Interval SUS Problem}
Hu et al.~\cite{Hu2014SUS} were the first to consider the interval SUS problem.
They proposed a data structure
answering a query returning all SUSs for the respective query interval in $O(\occ)$ optimal time
after $O(n)$ time preprocessing.
In the compressed setting,
Mieno et al.~\cite{Mieno2016SUSonRLE} considered the interval SUS problem
when the input string~$T$ is given \emph{run-length encoded (RLE)},
and proposed a data structure of size $O(r)$ words
answering a query by returning all SUSs for the respective query interval
in $O(\sqrt{\log r/\log\log r} + \occ)$ time, where $r$ is the number of single character runs in $T$.

\paragraph*{\bf Our Contribution} In this paper, we propose the following two data structures:
\begin{itemize}
    \item[(A)] A data structure of size $2n+2m+o(n)$ bits
               answering an interval SUS query in $O(\occ)$ time,
               where $m$ is the number of \emph{minimal unique substrings}
           of the input string\footnote{We show later in Lemma~\ref{lem:size_of_mus} that the number of minimal unique substrings~$m$ is at most $n$.},
               and $\occ$ is the number of SUSs of $T$ for the respective query interval
               (Theorem~\ref{thm:intervalSUS_DS}).
    \item[(B)] A data structure of size $\lceil(\log_2{3}+1)n\rceil+o(n)$ bits
               answering a point SUS query in $O(\occ)$ time,
               where $\occ$ is the number of SUSs of $T$ for the respective query point
               (Theorem~\ref{thm:pointSUS_compact}).
\end{itemize}

Instead of outputting the answer as a list of substrings of~$T$, it is sometimes sufficient to
output only the intervals corresponding to the respective substrings.
In such a case,
both data structures can answer a query \emph{without} the need of the input string.
The data structure (A) is the first data structure of size $O(n)$ bits
for the interval SUS problem.
Also, the data structure (B) is the first data structure of size $O(n)$ bits
for the point SUS problem, returning \emph{all} SUSs for a given query position.
Notice that the data structure of Ganguly et al.~\cite{Ganguly2017SUS} uses $4n+o(n)$ bits of space,
but returns only one SUS for a point SUS query.

Note that parts of this work have already been presented at
26th International Symposium of String Processing and Information Retrieval (SPIRE 2019)~\cite{mieno19sus}.
 \section{Preliminaries} \label{sec:pre}

Our model of computation is the word RAM with machine word size $\Omega(\log n)$.

\subsection{Strings} \label{subsec:notation}
Let $\Sigma$ be an alphabet.
An element of $\Sigma^*$ is called a \emph{string}.
For $|\Sigma| = 2$, we call a string also a bit array.
The length of a string $T$ is denoted by $|T|$.
The empty string $\varepsilon$ is the string of length 0.
Given a string $T$, the $i$-th character of $T$ is denoted by $T[i]$, for an integer~$i$ with $1 \leq i \leq |T|$.
For two integers~$i$ and $j$ with $1 \leq i \leq j \leq |T|$,
a substring of $T$ starting at position $i$ and ending at position $j$ is
denoted by $T[i..j]$. Namely, $T[i..j] = T[i]T[i+1]\cdots T[j]$.
For two strings $T$ and $w$,
the number of occurrences of $w$ in $T$ is denoted by $\numocc T(w) := |\{i \mid T[i..i+|w|-1] = w\}|$.
For two intervals $[i, j]$ and $[x, y]$,
let $\icover([i,j], [x,y]) := [\min\{i,x\}, \max\{j, y\}]$ denote the shortest interval
that contains the text positions $i, j, x,$ and $y$.
If the interval $[x,y]$ consists of a single point, i.e., $x = y$,
$\icover([i,j], [x,y])$ is denoted by $\icover([i,j], x)$ when we want to emphasize on the fact that $x=y$.

In what follows, we fix a string $T$ of length $n \ge 1$ whose characters are drawn from an integer alphabet $\Sigma$ of size $\sigma = n^{O(1)}$.

\subsection{MUSs and SUSs} \label{subsec:MUSandSUS}
Let $u$ be a non-empty substring of $T$.
$u$ is called a \emph{repeating substring} of $T$ if $\numocc T(u) \geq 2$,
and $u$ is called a \emph{unique substring} of $T$ if $\numocc T(u) = 1$.
Since every unique substring $u = T[i..j]$ of $T$ occurs exactly once in $T$,
we identify $u$ with its corresponding interval $[i, j]$.
We also say that the interval $[i, j]$ is unique
iff the corresponding substring $T[i..j]$ is a unique substring of $T$.

A unique substring $u = T[i..j]$ of $T$ is said to be a \emph{minimal unique substring} (\emph{MUS}) of $T$
iff every proper substring of $u$ is a repeating substring,
i.e., $\numocc T(T[i'..j']) \geq 2$ for every integer $i'$ and every integer $j'$
with $[i',j'] \subset [i,j]$ and $j'-i' < j-i$.
Let $\MUS_T :=  \{[i,j] \mid \mbox{$T[i..j]$ is a MUS of $T$}\}$ be the set of all intervals corresponding to the MUSs of $T$.
From the definition of MUSs, the next lemma follows:

\begin{lemma}[{\cite[Lemma~2]{Tsuruta2014SUS}}]\label{lem:size_of_mus}
    No element of $\MUS_T$ is nested in another
    element of $\MUS_T$, i.e.,
    two different MUSs $[i, j], [k, l] \in \MUS_T$ satisfy
    $[i, j] \not \subset [k, l]$ and $[k, l] \not \subset [i, j]$.
    Therefore, $0 < |\MUS_T| \leq |T|$.
\end{lemma}

We use the following two sets containing interval and point SUSs, which were defined at the beginning of the introduction:
Given an interval $[s,t] \subset [1,n]$,
$\SUS_T([s,t])$ denotes the set of the interval SUSs of $T$ for the interval $[s, t]$.
Given a text position $p \in [1,n]$,
$\SUS_T(p)$ denotes the set of the point SUSs of $T$ for the point $p$.

Given a query position $p \in [1,n]$ (resp. a query interval $[s,t] \subset [1, n]$),
the \emph{point} (resp. \emph{interval}) \emph{SUS~problem} is to compute $\SUS_T(p)$ (resp. $\SUS_T([s, t])$).
See Fig.~\ref{fig:MUSandSUS} for an example depicting MUSs and SUSs.

\begin{figure}[t]
    \centerline{\includegraphics[width=0.8\linewidth]{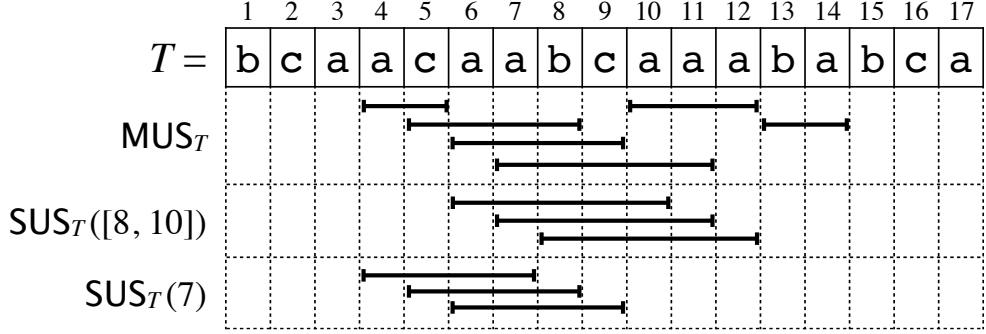}}
    \caption{
        The string $T = \mathtt{bcaacaabcaaababca}$,
        and its set $\MUS_T =$ $\{[4,5],$ $ [5,8],$ $ [6,9],$ $ [7,11],$ $ [10,12],$ $ [13,14]\}$.
        $\MUS_T$ corresponds to the set
        $\{\mathtt{ac},$ $ \mathtt{caab},$ $ \mathtt{aabc},$ $ \mathtt{abcaa},$ $ \mathtt{aaa},$ $ \mathtt{ba}\}$
        of all MUSs of $T$.
        The substrings
        $T[6..10] = \mathtt{aabca}$,
        $T[7..11] = \mathtt{abcaa}$, and
        $T[8..12] = \mathtt{bcaaa}$ are SUSs for the query interval $[8,10]$.
        Also, the substrings
        $T[4..7] = \mathtt{acaa}$,
        $T[5..8] = \mathtt{caab}$, and
        $T[6..9] = \mathtt{aabc}$ are SUSs for the query position $7$.
        The later defined leftmost/rightmost SUS and MUS (cf.\ Sect.~\ref{sec:pointSUS}) for $p = 7$ are
        $\protect\lmSUS_T^p = [4, 7]$,
        $\protect\lmMUS_T^p = [4, 5]$,
        $\protect\rmSUS_T^p = [6, 9]$, and
        $\protect\rmMUS_T^p = [6, 9]$.
    }\label{fig:MUSandSUS}
\end{figure}
 \section{Tools} \label{sec:tools}
In this section, we introduce the data structures needed for our approach to solve both SUS problems.

\subsection{Rank and Select}
Given a string $X$ of length $n$ over the alphabet $[1, \sigma]$.
For an integer $i$ with $1 \leq i \leq n$ and a character $c \in [1,\sigma]$,
the rank query $\rank_X(c, i)$ returns the number of the character $c$
in the prefix $X[1..i]$ of $X$.
Also, the select query $\select_X(c, i)$ returns the position of $X$
containing the $i$-th occurrence of the character $c$
(or returns the invalid symbol~$\nil$ if such a position does not exist).
For $\sigma = 2$ (i.e., $X$ is a bit array), we can make use of the following lemma:

\begin{lemma}[{\cite{Jacobson1989rank,Clark1998select}}]\label{lemRank}
We can endow a bit array~$X$ of length $n$ with a data structure answering $\rank_X$ and $\select_X$ in constant time.
This data structure takes $o(n)$ bits of space, and can be built on $X$ in $O(n)$ time with $O(\log n)$ bits of additional working space.
\end{lemma}

\subsection{Predecessor and Successor}\label{secPredecessor}
Let $Y$ be an array of length $k$ whose entries are positive integers in strictly increasing order.
Further suppose that these integers are less than or equal to $n$.
Given an integer $d$ with $1 \leq d \leq n$,
the \emph{predecessor} and the \emph{successor query} on $Y$ with $d$ are defined as
$\Pred_Y(d) := \max\{ i \mid Y[i] \leq d\}$ and
$\Succ_Y(d) := \min\{ i \mid Y[i] \geq d\}$, where we stipulate that $\min\{\} = \max\{\} = \nil$.

Let $BIT_Y$ be a bit array of length $n$ marking all integers present in $Y$, i.e.,
$BIT_Y[i]=\mathtt{1}$ iff there is an integer $j$ with $1 \le j \le k$ and $Y[j] = i$,
for every $i$ with $1 \leq i \leq n$.
By endowing $BIT_Y$ with a rank/select data structure,
we yield an $n+o(n)$ bits data structure answering
$\Pred_Y(d) = \select_{BIT_Y}(\mathtt{1}, \rank_{BIT_Y}(\mathtt{1},d))$
and $\Succ_Y(d)$\footnote{$\Succ_Y(d)$ can be computed similarly by considering the case whether $BIT_Y[d] =\mathtt{1}$.}
in constant time for each $d$ with $1 \leq d \leq n$.

\subsection{RmQ and RMQ}
Given an integer array $Z$ of length $n$ and an interval $[i,j] \subset [1,n]$,
the range minimum query $\RmQ_Z(i, j)$ (resp.\ the range maximum query $\RMQ_Z(i, j)$)
asks for the index~$p$ of a minimum element (resp.\ a maximum element) of
the subarray $Z[i..j]$, i.e.,
$p \in \arg\min_{i\leq k \leq j} Z[k]$, or respectively $p \in \arg\max_{i\leq k \leq j} Z[k]$.
We use the following well-known data structure to handle these kind of queries:
\begin{lemma}[\cite{Davoodi2012LCA}]\label{lemRMQ}
    Given an integer array $Z$ of length $n$,
    there is an $\RmQ$ (resp.\ $\RMQ$) data structure taking $2n+o(n)$ bits of space that can answer an $\RmQ$ (resp.\ $\RMQ$) query on $Z$ in constant time.
    This data structure can be constructed in $O(n)$ time with $o(n)$ bits of additional working space.
\end{lemma}

\subsection{Suffix Array and some Related Arrays}
We define six integer arrays
$\SA_T[1..n]$, $\ISA_T[1..n]$, $\LCP_T[1..n+1]$,
$\PLCP_T[1..n]$, $\RankPrev_T[1..n]$ and $\RankNext_T[1..n]$.
The suffix array $\SA_T$ of $T$ is the array with the property that
$T[\SA_T[i]..n]$ is lexicographically smaller than $T[\SA_T[i+1]..n]$
for every $i$ with $1 \leq i \leq n-1$~\cite{manber93sa}.
The inverse suffix array $\ISA_T$ of $T$ is the inverse of $\SA_T$, i.e.,
$\SA_T[\ISA_T[i]] = i$ for every $i$ with $1 \leq i \leq n$.
The longest common prefix array $\LCP_T$ of $T$ is the array with the property that
$\LCP_T[1] = \LCP_T[n+1] = 0$ and
$\LCP_T[i] = lcp(T[\SA_T[i]..n], T[\SA_T[i-1]..n])$ for every $i$ with $2 \leq i \leq n$,
where $lcp(P, Q)$ denotes the length of the longest common prefix of $P$ and $Q$
for two given strings $P$ and $Q$.
The permuted LCP array $\PLCP_T$ of $T$ is the array storing the values of $\LCP_T$
in text position order (instead of suffix array order), i.e.,
$\PLCP_T[i] = \LCP_T[\ISA_T[i]]$ for every $i$ with $1 \leq i \leq n$.
The rank predecessor array $\RankPrev_T$ of $T$ is the array with the property that
$\RankPrev_T[\SA_T[1]] = \nil$ and
$\RankPrev_T[\SA_T[i]] = \SA_T[i-1]$ for every $i$ with $2 \le i \le n$.
This array is also known as the $\Phi$ array in literature (e.g., \cite{Karkkainen2009PLCP, Goto2014LZ}).
The rank successor array $\RankNext_T$ of $T$ is the array with the property that
$\RankNext_T[SA[n]] = \nil$ and
$\RankNext_T[SA_T[i]] = SA_T[i+1]$ for every $i$ with $1 \le i \le n-1$.

\section{Computing MUSs in Compact Space} \label{sec:MUScompact}
For computing SUSs efficiently, it is advantageous to have a data structure available
that can retrieve MUSs starting or ending at specific positions, as the following lemma gives a crucial connection between MUSs and SUSs:

\begin{lemma}[{\cite[Lemma~2]{Tsuruta2014SUS}}] \label{lemOneMUS}
    Every point SUS contains exactly one MUS.
\end{lemma}

Fig.~\ref{fig:overview}  gives an overview of our introduced data structure and
shows the connections between this section and the following sections that focus on our two SUS problems.
For our data structure retrieving MUSs, we propose a compact representation and
an algorithm to compute this representation space-efficiently.
Our data structure is based on the following two bit arrays $\MB_T$ and $\ME_T$ of length $n$ with the properties that
\begin{itemize}
    \item $\MB_T[i] = \mathtt{1}$ iff $i$ is the beginning position of a MUS, and
    \item $\ME_T[i] = \mathtt{1}$ iff $i$ is the ending position of a MUS.
\end{itemize}
For the rest of this paper, let $m$ be the number of MUSs in $T$.
We rank the MUSs by their starting positions in the text, such that
the $j$-th MUS starts before the $(j+1)$-th MUS, for every integer~$j$ with $1 \leq j \leq m-1$.

Since MUSs are not nested (see Lemma~\ref{lem:size_of_mus}), the number of $\mathtt{1}$'s in
$\MB_T$ and $\ME_T$ is exactly $m$.
Hence,
the starting position, the ending position, and the length
of the $j$-th MUS can be computed with rank/select queries
for every integer~$j$ with $1 \leq j \leq m$.
How $\MB_T$ and $\ME_T$ can be computed is shown in the following lemma:

\begin{lemma}\label{lemBuildBitVectors}
    Let $\mathcal{D}_T$ be a data structure that can access $\ISA_T[i]$ and $\LCP_T[i]$ in $\pi_a(n)$ time
    for every position~$i$ with $1 \leq i \leq n$.
    Suppose that we can construct it in $\pi_c(n)$ time with $\pi_s(n)$ bits of working space
    including the space for $\mathcal{D}_T$.
    Then $\MB_T$ and $\ME_T$ can be computed
    in $O(\pi_c(n) + n\cdot \pi_a(n))$ total time
    while using $2n + \pi_s(n)$ bits of total working space
    including the space for $\MB_T$ and $\ME_T$.
\end{lemma}
\begin{proof}
    Given a text position~$i$ with $1 \leq i \leq n$,
    $T[i..i+\ell_i-1]$ with $\ell_i = \max\{\LCP_T[\ISA_T[i]], \LCP_T[\ISA_T[i]+1]\}$
    is the longest repeating substring starting at~$i$.
    If we extend this substring by the character to its right,
    it becomes unique.
    Thus, $T[i..i+\ell_i]$ is the shortest unique substring starting at $i$, except for the case
    that $i+\ell_i-1 = n$ as we cannot extend it to the right (hence, there is no unique substring starting at $i$ in this case).
    Additionally, the substring $T[i..i+\ell_i]$ is a MUS iff $T[i+1..i+\ell_{i}]$ is not unique
    (we already checked that $T[i..i+\ell_i-1]$ is not unique).
    $T[i+1..i+\ell_{i}]$ is not unique iff $\ell_i \leq \ell_{i+1}$
    since $T[i+1..i+1+\ell_{i+1}]$ is the shortest unique substring starting at $i+1$.
    Since each $\ell_i$ can be computed in $O(\pi_a(n))$ time for every $1\leq i \leq n$,
    the starting and ending positions of all MUSs (and hence, $\MB_T$ and $\ME_T$) can be computed
    in $O(n\cdot \pi_a(n))$ time by a linear scan of the text.
    Therefore, the total computing time is $O(\pi_c(n) + n\cdot \pi_a(n))$
    and the 
    total working space is $2n + \pi_s(n)$ bits including the space for $\MB_T$ and $\ME_T$.
\end{proof}

\begin{figure}[ht]
    \centerline{\includegraphics[width=0.8\linewidth]{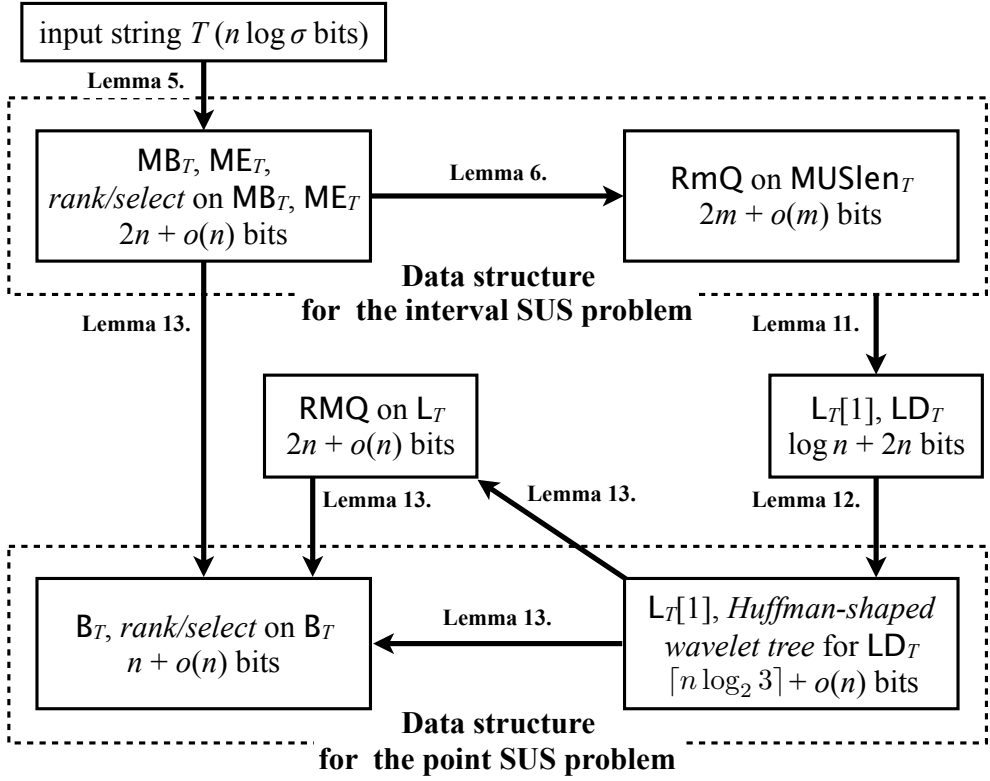}}
    \caption{Overview of the data structures proposed for solving the interval SUS and point SUS problem.
        Nodes are data structures.
        Edges of the same label (labeled by a certain lemma) describe an algorithm taking a set of input data structures to produce a data structure.
    }\label{fig:overview}
\end{figure}
 \section{Compact Data Structure for the Interval SUS Problem}\label{sec:intervalSUS}

In this section, we propose a compact data structure for the interval SUS problem.
It is based on the data structure of Mieno et al.~\cite{Mieno2016SUSonRLE},
which we review in the following.
We subsequently provide a compact representation of this data structure.

\paragraph*{\bf Data structures}
The data structure proposed by Mieno et al.~\cite{Mieno2016SUSonRLE}
consists of three arrays, each of length $m$:
$\X_T$, $\Y_T$, and $\MUSlen_T$.
The arrays $\X_T$ and $\Y_T$ store, respectively,  the beginning positions and ending positions of all MUSs
sorted by their beginning positions
such that the interval $[\X_T[i], \Y_T[i]]$ is the $i$-th MUS,
for every integer~$i$ with $1 \leq i \leq m$.
Further, $\MUSlen_T[i] = \Y_T[i] - \X_T[i] + 1$ stores the length of $i$-th MUS.
During a preprocessing phase, $\X_T$ and $\Y_T$ are endowed with
a successor and a predecessor data structure, respectively.
Further, $\MUSlen_T$ is endowed with an $\RmQ$ data structure.

\paragraph*{\bf Answering queries}
Given a query interval $[s, t]$,
let $\ell = \Pred_{Y_T}(t)$ be the index in~$Y_T$ of the largest ending position of a MUS that is at most $t$,
and $r = \Succ_{X_T}(s)$ be the index in~$X_T$ of the smallest starting position of a MUS that is at least $s$.
Then, $\SUS_T([s,t]) \subset \{\icover([s,t], [X_T[i], Y_T[i]]) \mid \ell \leq i \leq r\}$.
That is because the shortest intervals in
$\{\icover([s,t], [X_T[i], Y_T[i]]) \mid \ell \leq i \leq r\}$ correspond to the shortest unique substrings (SUSs)
among all substrings covering the interval $[s, t]$.
Thus, one of the SUSs for $[s, t]$ can be detected by considering
$\icover([s,t], [X_T[\ell], Y_T[\ell]])$ (as a candidate for the leftmost SUS),
$\icover([s,t], [X_T[r], Y_T[r]])$ (as a candidate for the rightmost SUS), and
$\RmQ_{\MUSlen_T}(\ell+1, r-1)$.
To output all SUSs, it is sufficient to answer $\RmQ$ queries on
subintervals of $\MUSlen_T[\ell+1..r-1]$ recursively.
In detail, suppose that there is a MUS in $\MUSlen_T[\ell+1..r-1]$ that is a SUS for $[s, t]$.
Further suppose that this is the $j$-th MUS having length~$k$.
Then we query $\MUSlen_T[\ell+1..j-1]$ and $\MUSlen_T[j+1..r-1]$ for all other MUSs of minimal length $k$.

\paragraph*{\bf Compact Representation}
Having the two bit arrays $\MB_T$ and $\ME_T$ of Section~\ref{sec:MUScompact},
we can simulate the three arrays $\X_T$, $\Y_T$, and $\MUSlen_T$.
By endowing these two bit arrays with rank/select data structures of Lemma~\ref{lemRank},
we can compute rank/select in constant time, which allows us to compute the value of
$\X_T[p]$, $\Y_T[p]$, $\MUSlen_T[p]$, $\Pred_{Y_T}(q)$ and $\Succ_{X_T}(q)$
for every index~$p$ with $1 \leq p \leq m$ and every text position~$q$ with $1\leq q \leq n$ in constant time while
using only $2n+o(n)$ bits of total space.
By endowing $\MUSlen_T$ with the $\RmQ$ data structure of Lemma~\ref{lemRMQ}, we can answer an $\RmQ$ query on $\MUSlen_T$ in constant time. This data structure takes $2m+o(m)$ bits of space.
Altogether, with these data structures we yield the following theorem:

\begin{theorem}\label{thm:intervalSUS_DS}
    For the interval SUS problem,
    there exists a data structure of size
    $2n + 2m + o(n)$ bits that
can answer an interval SUS query in $O(\occ)$ time,
    where $\occ$ is the number of SUSs of $T$ for the respective query interval.
\end{theorem}

Also, the data structure can be constructed
space-efficiently:

\begin{lemma}\label{lem:construct_intervalSUS_DS}
    Given $\MB_T$ and $\ME_T$,
    the data structure proposed in Theorem \ref{thm:intervalSUS_DS}
    can be constructed in $O(n)$ time using $2m + o(n)$ bits of total working space,
    which includes the space for this data structure.
\end{lemma}
\begin{proof}
    The data structure proposed in Theorem \ref{thm:intervalSUS_DS}
    consists of the two bit arrays $\MB_T$, $\ME_T$, and
    an $\RmQ$ data structure on $\MUSlen_T$, which is simulated by rank/select data structures on $\MB_T$ and $\ME_T$.
    Since $\MB_T$ and $\ME_T$ are already given,
    it is left to endow  $\MB_T$ and $\ME_T$ with rank/select data structures (using Lemma~\ref{lemRank}), and to compute the $\RmQ$ data structure on $\MUSlen_T$ (using Lemma~\ref{lemRMQ}).
\end{proof}
 \section{Compact Data Structure for the Point SUS Problem}\label{sec:pointSUS}
Before solving the point SUS problem, 
we borrow some additional notations from Tsuruta et al.~\cite{Tsuruta2014SUS} to deal with point SUS queries.
This is necessary since some of the MUSs never take part in finding a SUS such that
there is no meaning to compute and store them. 
Since we want to provide an output-sensitive algorithm answering a query in optimal time,
we only want to store MUSs that are candidates for being a SUS.

We say that the interval $[x,y] \in \MUS_T$ is a \emph{meaningful} MUS
if $T[x.. y]$ is a substring of (or equal to) a point SUS,
i.e., $\pcover([x,y], p) \in \SUS_T(p)$ for a position $p$.
Also, we say that the interval $[x,y] \in \MUS_T$ is a \emph{meaningless} MUS
if $[x,y]$ is not a meaningful MUS.
Let 
\begin{align*}
    \MMUS_T := \{[i,j] \in \MUS_T \mid& \mbox{ there exists a~} p \text{~with~} 1 \leq p \leq n   \\
                           & \mbox{ such that~} \pcover([i,j], p) \in \SUS_T(p)\}
\end{align*}
denote the set of all meaningful MUSs of $T$.

Let $\lmSUS_T^p$ denote the interval in $\SUS_T(p)$ with the leftmost starting position,
and let $\lmMUS_T^p$ denote the MUS contained in $\lmSUS_T^p$.
We say that $\lmSUS_T^p$ is the \emph{leftmost SUS} for $p$,
and $\lmMUS_T^p$ is the \emph{leftmost MUS} for $p$.
Similarly, we define the \emph{rightmost SUS} $\rmSUS_T^p$ and the \emph{rightmost MUS} $\rmMUS_T^p$ for $p$ by symmetry.
See Fig. \ref{fig:MUSandSUS} for an example for the leftmost/rightmost SUS and MUS.

Let $\length_T$ be an array of length $n$ such that
    $\length_T[i]$ is the length of a SUS\footnote{Although there can be multiple SUSs containing $i$, their lengths are all equal.} of $T$ containing~$i$
    for each position~$i$ with $1 \leq i \leq n$.
    Let $\MUSbegin_T$ be a bit array of length $n$ such that
    $\MUSbegin_T[i] = \mathsf{1}$ iff
    $i$ is the beginning position of a meaningful MUS of $T$.

From the definition of $\length_T$, we yield the following observation:

\begin{observation}\label{obs:Larray}
    For every position $p$ with $1 \leq p \leq n$ and every interval $[x, y] \in \SUS_T(p)$,
    $p-\length_T[p]+1 \leq x \leq p \leq y \leq p+\length_T[p]-1$.
\end{observation}

Next, we define the following four functions related to $\length_T$ and $\MUSbegin_T$.
    For a position~$q$ with $1 \leq q \leq n$ let
    \begin{itemize}
        \item $\predonepos_{\MUSbegin_T}(q) := \max\{ i \mid i \leq q \mbox{~and~} \MUSbegin_T[i] = \mathsf{1}\}$,
        \item $\succonepos_{\MUSbegin_T}(q) := \min\{ i \mid i \geq q \mbox{~and~} \MUSbegin_T[i] = \mathsf{1}\}$,
        \item $\predneq_{\length_T}(q) := \max\{ i \mid i < q \mbox{~and~} \length_T[i] \neq \length_T[q]\}$, and
        \item $\succneq_{\length_T}(q) := \min\{ i \mid i > q \mbox{~and~} \length_T[i] \neq \length_T[q]\}$.
    \end{itemize}
For all four functions, we stipulate that $\min \{\} = \max \{\} = \nil$.
See Fig.~\ref{fig:length} for an example of the arrays and functions defined above.

\begin{figure}[tbp]
\centerline{\includegraphics[width=0.8\linewidth]{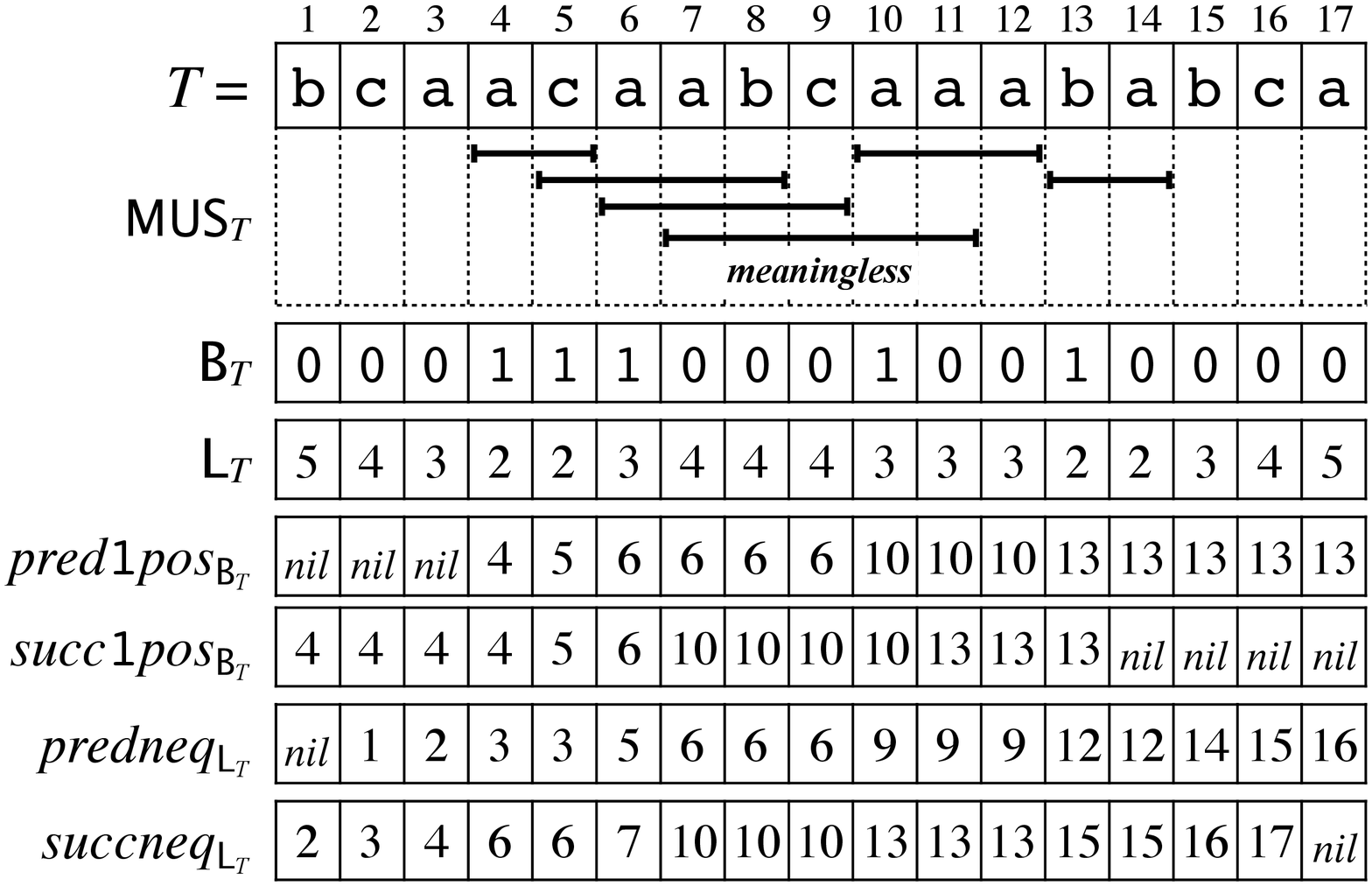}}
\caption{
    $\MUS_T, \MUSbegin_T$, $\length_T$, and the four functions defined in at the beginning of Section~\ref{sec:pointSUS}
    for the string $T = \mathtt{bcaacaabcaaababca}$.
    $\MUSbegin_T[7] = \mathtt{0}$
    because the MUS $T[7.. 11] = \mathtt{abcaa}$ is meaningless.
}\label{fig:length}
\end{figure}

\subsection{Finding SUSs with $\length$ and $\MUSbegin$}\label{subsec:pointSUS_Query}

Our idea is to answer point SUS queries with $\length_T$ and $\MUSbegin_T$.
For that, we first think about how to find the leftmost and rightmost SUS for a given query (Observation~\ref{obs:Larray} gives us the range in which to search).
Having this leftmost and the rightmost SUS, we can find all other SUSs with $\MUSbegin_T$ marking the
beginning positions of the meaningful MUSs that correspond to the SUSs we want to output.
Before that, we need some properties of $\length_T$ that help us to prove the following lemmas in this section:
Lemma~\ref{lem:Larray_diff} gives us a hint on the shape of $\length_T$, while
Lemma~\ref{lem:Lupdown} shows us how to find SUSs based on two consecutive values of $\length_T$ with a connection to MUSs.

\begin{lemma}\label{lem:Larray_diff}
    $|\length_T[p] - \length_T[p+1]| \leq 1$ for every position~$p$ with $1 \leq p \leq  n-1$.
\end{lemma}
\begin{proof}
    Let $\ell = \length_T[p]$ and $\ell' = \length_T[p+1]$.
    From the definition of $\length_T$,
    there exists a unique substring of length $\ell$ containing the position $p$.
    If $\ell < \ell'$, there is no unique substring of length $\ell$ containing $p+1$.
    Thus, $T[p-\ell+1..p]$ is unique, and consequently $T[p-\ell+1..p+1]$ is also unique.
    Hence, $\ell' = \ell+1$.
    Similarly, in the case of $\ell > \ell'$, it can be proven that $\ell' = \ell-1$.
\end{proof}

\begin{lemma}\label{lem:Lupdown}
    Let $p$ be a position with $1 \leq p \leq |T|-1$, and let $\ell := \length_T[p]$.
    If $\length_T[p+1] = \ell + 1$, then
    \begin{itemize}
        \item $T[p-\ell+1.. p  ] \in \SUS_T(p)$,
        \item $T[p-\ell+1.. p+1] \in \SUS_T(p+1)$, and
        \item $p-\ell+1$ is the starting position of a MUS of $T$.
    \end{itemize}
    If $\length[p+1] = \ell - 1$ then
    \begin{itemize}
	\item $T[p  .. p+\ell-1] \in \SUS_T(p)$,
	\item $T[p+1.. p+\ell-1] \in \SUS_T(p+1)$, and
	\item $p+\ell-1$ is the ending position of a MUS of $T$.
    \end{itemize}
\end{lemma}
\begin{proof}
    First, we consider the case that $\length_T[p+1] = \ell + 1$.
    From the proof of Lemma~\ref{lem:Larray_diff},
    $T[p-\ell+1..p]$ and $T[p-\ell+1..p+1]$ are unique substrings in $T$.
    Thus, $T[p-\ell+1.. p] \in \SUS_T(p)$ and $T[p-\ell+1.. p+1] \in \SUS_T(p+1)$.
    Since every point SUS contains exactly one MUS (cf.\ Lemma~\ref{lemOneMUS}), 
    there exists a MUS $[b, e] \subset [p-\ell+1, p]$.
    Assume that $b > p-\ell+1$, then $T[b..p]$ is the shortest unique substring 
    among all substrings containing the text position~$p$.
    Its length is $p-b+1 < \ell$.
    This contradicts that $T[p-\ell+1.. p] \in SUS_T(p)$, and therefore $b = p-\ell+1$ must hold.
    The remaining case $\length_T[p+1] = \ell - 1$ can be proven analogously by symmetry.
\end{proof}

In the following two lemmas (Lemmas~\ref{lem:lmSUS} and~\ref{lem:rmSUS}), we focus on finding the leftmost SUS and the rightmost SUS for a given query point.
That is because the leftmost SUS and the rightmost SUS give us an interval containing the starting positions of the remaining SUSs we want to report\footnote{The actual reporting of those SUSs is done in Lemma~\ref{lem:leftmost_rightmost}}.
 
\begin{figure}[ht]
    \centerline{\includegraphics[width=120mm]{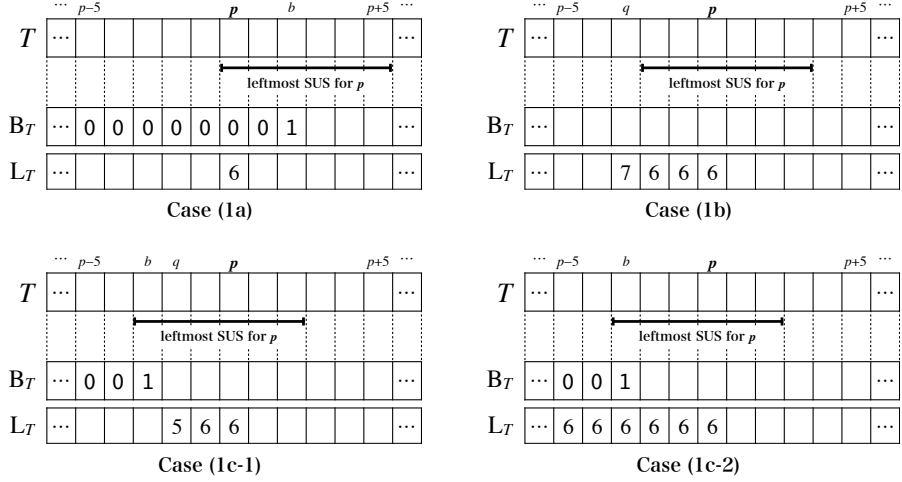}}
    \caption{Example of the proof of Lemma~\ref{lem:lmSUS}
        with $\length_T[p] = \ell = 6$.
        The example (as well as all later examples in this section) still works when replacing the number $\ell = 6$ with another number as long as the
        relative differences to the other entries in $\length_T$ and the search range in $T$ is kept.
    }\label{fig:lmSUS}
\end{figure}

\begin{lemma}\label{lem:lmSUS}
    Let $p$ be a position with $1 \leq p \leq n$, and let
$\ell = \length_T[p]$,
    $b = \succonepos_{\MUSbegin_T}(\max\{1,p-\ell+1\})$,
    and
    $q = \predneq_{\length_T}(p)$.
    Then, $b \leq \min\{p+\ell-1, n\}$ and
    \begin{subnumcases}{\lmSUS_T^p=}
        {}[p,p+\ell-1] & if $b \geq p$,\\
        {}[q+1,q+\ell] & if $b < p$, $q \geq p-\ell+1$, and $\length_T[q] > \ell$,\\
        {}[b,b+\ell-1] & otherwise.
    \end{subnumcases}
\end{lemma}
\begin{proof}
If $\ell = 1$, it is clear that the interval $[p, p]$ of length 1 is a MUS of $T$,
    thus $b = p$ and $\lmSUS_T^p = [p, p]$.
    For the rest of the proof, we focus on the case that $\ell \geq 2$.
Since $\length_T[p] = \ell$,
    there exists a unique substring of length $\ell$ containing the position $p$,
    and there exists at least one MUS that is a subinterval of $[p-\ell+1, p+\ell-1]$.
    Thus, $b \leq \min\{p+\ell-1, n\}$.
    See Fig.~\ref{fig:lmSUS} for an illustration of each of the above cases we consider in the following:
\begin{enumerate}
\item[(1a)]
            Assume that there exists a unique substring $T[p'.. p'+\ell-1]$
            containing the position $p$ with $p' < p$.
            Since $b \geq p > p'$, $T[p'+1.. p'+\ell-1]$ is also unique and contains position $p$.
            It contradicts $\length_T[p] = \ell$; therefore, $\lmSUS_T^p = [p, p+\ell-1]$.
\item[(1b)]
            From the definition of $q$ and Lemma~\ref{lem:Larray_diff},
            $\length_T[q] = \ell+1$ and $\length_T[q+1] = \ell$.
            From Lemma~\ref{lem:Lupdown}, $[q+1, q+\ell]$ is unique.
            Also, $[q+1, q+\ell] \in \SUS_T(p)$ because $p \in [q+1, q+\ell]$.
            Since $\length_T[q] = \ell+1$, there is no unique substring that
            contains the position $q$ and is shorter than $\ell+1$.
            Therefore, $\lmSUS_T^p = [q+1, q+\ell]$.
\item[(1c)]
            We divide this case into two subcases:
            \begin{itemize}
                \item[(1c-1)] $b < p$ and $q \geq p-\ell+1$ and $\length_T[q] < \ell$, or
                \item[(1c-2)] $b < p$ and $q <    p-\ell+1$.
            \end{itemize}
In Subcase~(1c-1),
            from the definition of $q$ and Lemma~\ref{lem:Larray_diff},
            $\length_T[q] = \ell-1$ and $\length_T[i] = \ell$ for all~$i \in [q+1, p]$.
            From Lemma~\ref{lem:Lupdown}, the interval $[q-\ell+2, q]$ of length $\ell-1$ is unique.
            Since $[p-\ell+1, q] \subset [q-\ell+2, q]$,
            $\length_T[i] \leq \ell-1$  for all~$i \in [p-\ell+1, q]$.
In Subcase~(1c-2),
            it is clear that $\length_T[i] = \ell$ for all $i \in [p-\ell+1, p]$.
Therefore, $\length_T[i] \leq \ell$ for all $i \in [p-\ell+1, p]$
            in both subcases.

Let $e$ be the ending position of the meaningful MUS $[b,e]$ starting at the position $b$,
            and $\ell' = e-b+1$ be the length of this MUS.
We assume $\ell' > \ell$ for the sake of contradiction
            (and thus $[b,e]$ cannot be $\lmMUS_T^p$ whose length is at most $\ell$).
            Since $b \geq p-\ell+1$ and $\ell' > \ell$, $e > p$ must hold.
            Let $[b', e'] = \lmMUS_T^p$.
Since (a) there is no interval $[x, y] \in \SUS_T(p)$ such that $x < \min\{b, p\}$,
            and (b) MUSs cannot be nested, it follows that $b' > b$ and $e' > e$.
            Thus, $\lmSUS_T^p = \pcover([b',e'], p) = [e'-\ell+1, e']$
            and $\length_T[i] \leq \ell$ for all $p \leq i \leq e'$.
            Since $[b, e] \subset [p-\ell+1, e']$,
            $\length_T[i] \leq \ell$ for all $b \leq i \leq e$.
            This contradicts that the MUS $[b, e]$ of length $\ell' > \ell$ is a meaningful MUS.
            Therefore, $\ell' \leq \ell$ and $\lmSUS_T^p = \pcover([b,e],p) = [b, b+\ell-1]$.
    \end{enumerate}
\end{proof}
From Lemma~\ref{lem:lmSUS} we yield the following corollary:

\begin{corollary}\label{col:lmSUS}
    If we can compute $\length_T[i]$,
    $\predneq_{\length_T}(i)$ and $\succonepos_{\MUSbegin_T}(i)$ in constant time
    for each~$i$ with $1 \leq i \leq n$,
    we can compute $\lmSUS_T^p$ in constant time for each position~$p$ with $1 \leq p \leq n$.
\end{corollary}

\begin{figure}[ht]
    \centerline{\includegraphics[width=120mm]{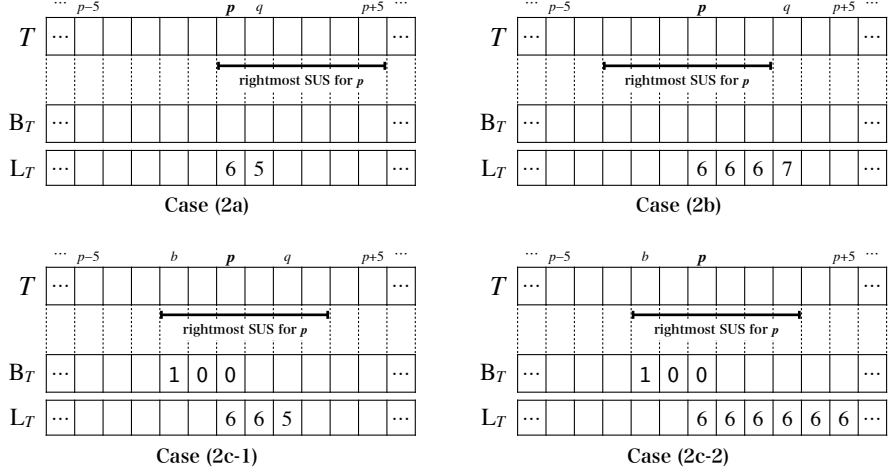}}
    \caption{Example of the proof of Lemma~\ref{lem:rmSUS}
        with $\length_T[p] = \ell = 6$.
    }\label{fig:rmSUS}
\end{figure}

\begin{lemma}\label{lem:rmSUS}
    Let $p$ be a position with $1 \leq p \leq n$, and
    let $\ell = \length_T[p]$, $q = \succneq_{\length_T}(p)$,
    and $b = \predonepos_{\MUSbegin_T}(p)$.
    Then,
    \begin{subnumcases}{\rmSUS_T^p=}
        {}[p,p+\ell-1] & if $q =    p + 1$     and $\length_T[q] < \ell$, \\
        {}[q-\ell,q-1] & if $q \leq p + \ell - 1$ and $\length_T[q] > \ell$, \\
        {}[b,b+\ell-1] & otherwise.
    \end{subnumcases}
\end{lemma}
\begin{proof}
If $\ell = 1$, it is clear that the interval $[p, p]$ of length 1 is a MUS of $T$,
    thus $b = p$ and $\rmSUS_T^p = [p, p]$. We consider the condition of $\ell \geq 2$.
    See Fig.~\ref{fig:rmSUS} for an illustration of each of the above cases we consider in the following:
    \begin{enumerate}
\item[(2a)]
            From Lemma~\ref{lem:Lupdown}, $[p, p+\ell-1]$ is a  SUS for $p$, which is by definition the rightmost one.
\item[(2b)]
            In this case, $\length_T[q] = \ell+1$ and $\length_T[q-1] = \ell$.
            From Lemma~\ref{lem:Lupdown}, $[q-\ell, q-1]$ is unique.
            Since $p \in [q-\ell, q+1]$, $[q-\ell, q-1]$ is a SUS for $p$.
            Additionally, there is no unique interval $[x, y] \in \SUS_T(p)$
            such that $y \geq q$ because $\length_T[q] = \ell+1$.
            Thus, $\rmSUS_T^p = [q-\ell, q-1]$.
\item[(2c)]
            We divide this case into two subcases:
            \begin{itemize}
                \item[(2c-1)] $p+1 < q \leq p+\ell-1$ and $\length_T[q] < \ell$, or
                \item[(2c-2)]       $q >    p+\ell-1$.
            \end{itemize}
In Subcase (2c-1),
            $\length_T[p+1] = \ell$ and $\length_T[q] = \ell-1$ and $\length_T[i] = \ell$
            for all $p+2\leq i \leq q-1$.
            From Lemma~\ref{lem:Lupdown}, $[q, q+\ell-2]$ (of length $\ell-1$) is unique.
            Since $[q, p+\ell-1] \subset [q, q+\ell-2]$, $\length_T[i] \leq \ell-1$
            for all $q \leq i \leq p+\ell-1$.
In Subcase (2c-2),
            from the definition of $q$, $\length_T[i] = \ell$ for all $p \leq i \leq p+\ell-1$.
Therefore, $\length_T[p+1] = \ell$ and $\length_T[i] \leq \ell$
            for all integers~$i$ with $p+2 \leq i \leq p+\ell-1$
            in both subcases.

For the sake of contradiction, assume that there is a MUS $[b', e']$
            such that $b' > p$ and $\pcover([b', e'], p) = [p, e'] \in \SUS_T(p)$.
            Since $\length_T[p] = \ell$, $[p, e']$ is a unique substring of length $\ell$.
            Hence, $\pcover([b',e'], p+1) = [p+1, e']$ is a unique substring of length $\ell-1$.
            It contradicts $\length_T[p+1] = \ell$;
            therefore, the beginning position of the rightmost MUS for $p$ is at most $p$.
Next, we show that the MUS starting at $b$ is the rightmost meaningful MUS for $p$.
            Let $e$ be its ending position, and $\ell' = e-b+1$ be its length.
We assume that $\ell' > \ell$ for the sake of contradiction (and thus, $[b,e]$ is not $\rmMUS_T^p$ whose length is at most $\ell$).
            Since $\length_T[p] = \ell$, $b \geq p-\ell+1$ and $e > p$.
            Let $[b'', e''] = \rmMUS_T^p$.
            Since MUSs cannot be nested, $b'' < b$.
            Since $e''-b''+1 \leq \ell$,
            $\length_T[i] \leq \ell$ for all $i$ with $b'' \leq i \leq p+\ell-1$.
We consider two cases to obtain a contradiction:
\begin{itemize}
                \item If $e \leq p+\ell-1$ then it is clear that $\length_T[i] \leq \ell$ for all $i$ with $b \leq i \leq e$.
                    This contradicts that the MUS $[b, e]$ of length $\ell'$ is a meaningful MUS.
                \item If $e > p+\ell-1$, it is clear that $|\pcover([b,e], p+\ell-1)| = |[b,e]| = \ell'$.
                    Since $\length_T[p+\ell-1] \leq \ell$,
                    there exists a unique substring $[s, t]$ such that
                    $s \leq p+\ell-1 \leq t$ and $t-s+1 \leq \ell$.
                    Hence, $\length_T[i] \leq \ell$ for all $i$ with $s \leq i \leq t$.
                    Since $[b, e]$ is a MUS and $p \leq  s$, $b < s < e < t$.
                    Consequently, $\length_T[i] \leq \ell$ for all $i$ with $b \leq i \leq e$ and
                    this contradicts that the MUS $[b, e]$ of length $\ell'$ is a meaningful MUS.
            \end{itemize}
Therefore, $\ell' \leq \ell$ and $\rmSUS_T^p = \pcover([b,e],p) = [b, b+\ell-1]$.
    \end{enumerate}
\end{proof}

\begin{corollary}\label{col:rmSUS}
    If we can compute $\length_T[i]$,
    $\succneq_{\length_T}(i)$ and $\predonepos_{\MUSbegin_T}(i)$ in constant time
    for each~$i$ with $1 \leq i \leq n$,
    we can compute $\rmSUS_T^p$ in constant time for each position~$p$ with $1 \leq p \leq n$.
\end{corollary}

\subsection{Compact Representations of $\length$}\label{subsec:Larray_Compact}

We now propose a succinct representation of the array $\length_T$
    consisting of the integer array $\lengthdiff_T$ of length $n$ defined as
    $\lengthdiff_T[1] = 0$ and $\lengthdiff_T[i] = \length_T[i] - \length_T[i-1] \in \{\mathtt{-1},\mathtt{0},\mathtt{1}\}$
    for every $i$ with $2 \le i \le n$.

\begin{lemma}\label{lem:build_tritvector}
    The data structure of Theorem~\ref{thm:intervalSUS_DS} 
    can compute $\length_T[p]$ in constant time with $O(\log n)$ bits of additional working space
    for each $p$ with $1\leq p \leq n$.
\end{lemma}
\begin{proof}
    Suppose that we have the data structure~$D$ of Theorem~\ref{thm:intervalSUS_DS} and want to know $\length_T[p]$.
    We query $D$ with the interval~$[p, p]$ to retrieve \emph{one} SUS for the query interval $[p,p]$ in constant time.
    This can be achieved by stopping the retrieval after the first SUS $[i, j] \in \SUS_T([p, p])$ has been reported.
    Since all SUSs for $[p, p]$ have the same length, $\length_T[p] = j-i+1$.
    The additional working space is $O(\log n)$ bits.
\end{proof}

\begin{table}[h]
  \begin{center}
  \begin{tabular}{|c|l|ll|} \hline
      \multicolumn{1}{|c|}{No.} &
      \multicolumn{1}{|c|}{Process} &
      \multicolumn{2}{|c|}{
          \begin{tabular}{c}
              Total working space in bits\\
              (excluding $\MB_T$ and $\ME_T$)
          \end{tabular}
      } \\ \hline
      1 & input $\MB_T$, $\ME_T$ & - & \\ \hline
      2 & construct $\RmQ$ on $\MUSlen_T$ & $2m + o(n)$ & Lemma \ref{lem:construct_intervalSUS_DS}\\ \hline
      3 & construct $\lengthdiff_T$, $\length_T[1]$ & $2n + 2m + o(n)$ & Lemma \ref{lem:build_tritvector}\\ \hline
      4 & free $\RmQ$ on $\MUSlen_T$ & $2n + o(n)$ & \\ \hline
      \multirow{2}{*}{5} & construct \emph{Huffman-shaped} &
      \multirow{2}{*}{$2n + \lceil n\log_2{3}\rceil + o(n)$} & \multirow{2}{*}{Lemma \ref{lem:tritvector_compact}}\\
      & \emph{Wavelet Tree} for $\lengthdiff_T$ & & \\ \hline
      6 & free $\lengthdiff_T$ & $\lceil n\log_2{3}\rceil + o(n)$ & \\ \hline
      7 & construct $\RMQ$ on $\length_T$ & $2n + \lceil n\log_2{3}\rceil + o(n)$ & Lemma \ref{lem:Barray_compact}\\ \hline
      8 & construct $\MUSbegin_T$ & $3n + \lceil n\log_2{3}\rceil + o(n)$ & Lemma \ref{lem:Barray_compact}\\ \hline
  \end{tabular}
  \end{center}
  \caption{Working space used during the construction of the data structure proposed in Theorem~\ref{thm:pointSUS_compact}.
    We can free up space of no longer needed data structures between several steps.
    See also Fig.~\ref{fig:overview} for the dependencies of the execution, and other possible ways to build the final data structure.
    However, these other ways need more maximum working space (at some step) than the way listed in this table.
  }
  \label{tab:pointSUS_compact}
\end{table}

Lemma~\ref{lem:build_tritvector} allows us to compute $\lengthdiff_T$ in $O(n)$ time,
which we represent as an integer array with bit width two, thus using $2n$ bits of space.
In the following, we build a compressed rank/select data structure on $\lengthdiff_T$.
This data structure is a self-index such that we no longer need to keep $\lengthdiff_T$ in memory.
With $\lengthdiff_T$ we can access $\length_T$, as can be seen by the following lemma:

\begin{lemma}\label{lem:tritvector_compact}
    There exists a data structure of size $\lceil n\log_2{3}\rceil + o(n)$ bits
    that can access $\length_T[i]$, and can compute $\predneq_{\length_T}(i)$
    and $\succneq_{\length_T}(i)$ in constant time for each position $i$ with $1 \leq i \leq n$.
    Given $\MB_T$ and $\ME_T$,
    the data structure can be constructed in $O(n)$ time
using $2n + \max\{\lceil n\log_2{3}\rceil, 2m\} + o(n)$ bits of total working space, which
    includes the space for this data structure.
\end{lemma}
\begin{proof}
     The following equations hold for every text position~$i$ with $1 \leq i \leq n$:
    \begin{eqnarray*}
        \length_T[i]            &=&
            \length_T[1] + \rank_{\lengthdiff_T}(\mathsf{1}, i) - \rank_{\lengthdiff_T}(\mathsf{-1}, i),\\
        \predneq_{\length_T}(i) &=&
            \max\{\select_{\lengthdiff_T}(c,\rank_{\lengthdiff_T}(c,i))-1 \mid c\in\{\mathsf{-1},\mathsf{1}\}\},\\
        \succneq_{\length_T}(i) &=&
            \min\{\select_{\lengthdiff_T}(c,\rank_{\lengthdiff_T}(c,i)+1) \mid c\in\{\mathsf{-1},\mathsf{1}\}\}.
    \end{eqnarray*}
    We can compute the value of $\length_T[1]$ and $\lengthdiff_T$ with Lemma~\ref{lem:build_tritvector}.
With a rank/select data structure on $\lengthdiff_T$
    we can compute the above functions.
    Such a data structure is the Huffman-shaped wavelet tree~\cite{Makinen2005SSA}.
    This data structure can be constructed in linear time and takes $\lceil n\log_2{3}\rceil + o(n)$ bits of space,
    since the possible number of different values in $\lengthdiff_T$ is three.
    Therefore, it can also provide answers to rank/select queries in constant time.
\end{proof}

Finally, we show how to compute $\MUSbegin_T$:

\begin{lemma}\label{lem:Barray_compact}
    There exists a data structure of size $n + o(n)$ bits
    that can compute
    $\succonepos_{\MUSbegin_T}(i)$ and $\predonepos_{\MUSbegin_T}(i)$
    in constant time for each position $1 \leq i \leq n$.
    Given $\MB_T$ and $\ME_T$,
    this data structure can be constructed in $O(n)$ time
using $3n + \lceil n\log_2{3}\rceil + o(n)$ bits of
    total working space including the space for this data structure.
\end{lemma}
\begin{proof}
    Our idea is to compute $\MUSbegin_T$ since
    the following equations hold for every text position~$i$ with $1 \leq i \leq n$ (cf.\ $BIT_Y$ in Section~\ref{secPredecessor}):

    \begin{eqnarray*}
    \predonepos_{\MUSbegin_T}(i) &=&
    \begin{cases}
        i \qquad &~\mbox{if }\MUSbegin_T[i]=\mathsf{1},\\
        \select_{\MUSbegin_T}(\mathsf{1},\rank_{\MUSbegin_T}(\mathsf{1},i)) \qquad &~\mbox{if }\MUSbegin_T[i]=\mathsf{0}.\\
    \end{cases}\\
    \succonepos_{\MUSbegin_T}(i) &=&
    \begin{cases}
        i ~&~\mbox{if }\MUSbegin_T[i]=\mathsf{1},\\
        \select_{\MUSbegin_T}(\mathsf{1},\rank_{\MUSbegin_T}(\mathsf{1},i)+1) ~&~\mbox{if }\MUSbegin_T[i]=\mathsf{0}.\\
    \end{cases}
    \end{eqnarray*}
In the following we show how to compute $\MUSbegin_T$ from $\MB_T$ and $\ME_T$
    in linear time with linear number of bits of working space.
    Let $b_i = \select_{\MB_T}(\mathtt{1}, i)$ and $e_i = \select_{\ME_T}(\mathtt{1}, i)$
    be the starting position and the ending position of the $i$-th MUS respectively,
    for each $1 \leq i \leq m$.
    Given $x_i = \RMQ_{\length_T}(b_i, e_i)$,
    $\length_T[x_i] \leq e_i-b_i+1$
    since $b_i \leq x_i \leq e_i$ and $[b_i, e_i]$ is unique.
    If $\length_T[x_i] < e_i-b_i+1$,
    there is no position $p$ with $\icover([b_i, e_i], p) \in \SUS_T(p)$, i.e.,
    $[b_i, e_i]$ is a meaningless MUS.
    Otherwise ($\length_T[x_i] = e_i-b_i+1$),
    $\icover([b_i, e_i], x_i) = [b_i, e_i] \in \SUS_T(x_i)$, i.e.,
    $[b_i, e_i]$ is a meaningful MUS.
    Hence, it can be detected in constant time
    whether a MUS is meaningful
    by an $\RMQ$ query on $\length_T$.
    We can compute the compact representation of $\length_T$ described in Lemma~\ref{lem:tritvector_compact}.
    The data structure takes $\lceil n\log_2{3}\rceil + o(n)$ bits and can be constructed with $2n + \max\{\lceil n\log_2{3}\rceil, 2m\} + o(n)$ bits of total working space.
    Subsequently, we endow it with the $\RMQ$ data structure of Lemma~\ref{lemRMQ}
    in $O(n)$ time using $2n + o(n)$ bits of space.
    Therefore, the computing time of $\MUSbegin_T$ is $O(n)$
    and the working space is, aside from the space for $\MB_T$ and $\ME_T$,
    $3n + \lceil n\log_2{3}\rceil + o(n)$ bits, including the space for $\MUSbegin_T$.
    Finally, we can endow $\MUSbegin_T$ with rank/select data structures,
    which allows us to compute each of the above two functions $\predonepos_{\MUSbegin_T}$ and $\succonepos_{\MUSbegin_T}$ in constant time.
\end{proof}

Actually, having $\MB_T$ and $\ME_T$ available,
we can simulate an access to $\length_T[i]$ in constant time with Lemma~\ref{lem:build_tritvector}
by using the $\RmQ$ data structure on $\MUSlen_T$.
This allows us to compute the $\RMQ$ data structure on $\length_T$ directly
without the need for computing $\MUSbegin_T$ in the first place,
i.e., we can replace the working space of Lemma~\ref{lem:Barray_compact} with $2n + 2m + o(n)$ additional bits of working space.
However, since our final data structure needs $\lengthdiff_T$,
computing $\MUSbegin_T$ before $\lengthdiff_T$ would require more working space in the end than the other way around,
since we no longer need the $\RmQ$ data structure on $\MUSlen_T$
after having built the rank/select data structure of Lemma~\ref{lem:tritvector_compact}.

Before stating our final theorem, we need a property for meaningful MUSs:
\begin{lemma} \label{lem:leftmost_rightmost}
    On the one hand, $\pcover([s_i, e_i], p) \in \SUS_T(p)$
    for every meaningful MUS $[s_i, e_i]$
    starting with or after the leftmost MUS for $p$ and starting before or with the rightmost MUS for $p$.
    On the other hand, each element (i.e., an interval) of $\SUS_T(p)$
    starting with or after the leftmost MUS for $p$ and starting before or with the rightmost MUS for $p$
    contains exactly one distinct MUS.
\end{lemma}
\begin{proof}
    The first part is shown by Tsuruta et al.~\cite[Lemma~3]{Tsuruta2014SUS}.
    The second part is due to Lemma~\ref{lemOneMUS}.
\end{proof}

\begin{theorem}\label{thm:pointSUS_compact}
    For the point SUS problem, there exists a data structure of size $n + \lceil n\log_2{3}\rceil + o(n)$ bits
    that can answer a point SUS query in $O(\occ)$ time,
    where $\occ$ is the number of SUSs of $T$ for the respective query point.
    Given $\MB_T$ and $\ME_T$,
    the data structure can be constructed in $O(n)$ time
using $3n + \lceil n\log_2{3}\rceil + o(n)$ bits of total working space, which
    includes the space for this data structure.
\end{theorem}
\begin{proof}
    Let $p$ be a query position, and suppose that the number of SUSs for $p$ is $\occ$.
    Like the MUSs in Section~\ref{sec:MUScompact}, we rank the SUSs for $p$ by their starting positions.
    Let $[s_j, e_j]$ be the $j$-th SUS for $p$ with $1 \le j \le \occ$
    such that $[s_1, e_1]$ and $[s_{\occ}, e_{\occ}]$ are the leftmost SUS and the rightmost SUS for~$p$, respectively.
    If $s_1 = p$ then $[s_1, e_1] = [s_{\occ}, e_{\occ}]$, and thus the output consists of this single interval.
    Otherwise ($s_1 \neq p$),
    we can compute $s_{i}$ iteratively from $s_{i-1}$
    by $s_i = \select_{\MUSbegin_T}(\mathtt{1}, \rank_{\MUSbegin_T}(\mathtt{1}, s_{i-1}) + 1)$
    in constant time for each $i$ with $2 \le i \le \occ-1$,
    allowing us to answer the query in time linear to the number of SUSs.
    As $\occ$ is not known in advance, we stop the iteration
    whenever we computed an $s_i$ that is larger than the starting position of the rightmost SUS for~$p$.
    A detailed analysis of the claimed working space is given in Table~\ref{tab:pointSUS_compact}.
\end{proof}

\begin{corollary}\label{col:num_of_SUSs}
  The data structure of Theorem~\ref{thm:pointSUS_compact}
  can compute the number of SUSs for a query position in constant time.
\end{corollary}
\begin{proof}
    Let $[s_l, e_l]$ and $[s_r, e_r]$ be the leftmost and the rightmost SUS for a given query position, respectively.
    All MUSs starting between $s_l$ and $s_r$ (excluding $s_l$ and $s_r$) are SUSs for this query position.
    Let $occ'$ be their number.
    Therefore, the number we want to output is $occ = occ'+2$.
    With Lemmas~\ref{lem:lmSUS} and~\ref{lem:rmSUS}, we can find $[s_l, e_l]$ and $[s_r, e_r]$ in constant time.
    Further, we can compute $occ'$ in constant time since $occ' = \rank_{\MUSbegin_T}(\mathtt{1}, s_r-1)-\rank_{\MUSbegin_T}(\mathtt{1}, s_l)$.
\end{proof}
\begin{table}
  \begin{center}
    \begin{tabular}{lllll}
      \toprule
      & & \multicolumn{2}{c}{Time}
      \\
      \cmidrule{3-4}
      data structure & Lemma & Access & Construction & Space in bits
      \\
      &  & $\pi_a(n)$                  & $\pi_c(n)$                     & $\pi_s(n)$ \\
      \midrule
      sparse suffix tree     & \ref{lemLCPSST} & $O(\epsilon^{-1})$          & $O(\epsilon^{-2} n)$           & $(1+\epsilon)n \log n + O(n)$ \\
      com\-pressed suffix tree & \ref{lemLCPCST} & $O(\log_\sigma^\epsilon n)$ & $O(n \log_\sigma^\epsilon n )$ & $O(\epsilon^{-1} n \log \sigma)$ \\
      RLBWT                  & \ref{lemLCPBWT} & $O(\log r \log^{O(1)} n)$   & $O(n \log r + r \log^{O(1)} n)$       & $r \log^{O(1)} n + O(n)$  \\
      \bottomrule
    \end{tabular}
    \caption{Efficient data structures with access to $\ISA_T[i]$ and $\LCP_T[i]$. $\epsilon$ is a constant with $0 < \epsilon \le 1$. $r$ denotes the number of single character runs in the BWT.}
    \label{tabLCPDS}
  \end{center}
\end{table}

\section{Auxiliary Data Structure }
In the following, we present three different representations of the data structure required by Lemma~\ref{lemBuildBitVectors}.
See Table~\ref{tabLCPDS} for a juxtaposition of their characteristics.
Also, we present an algorithm for computing MUSs by using $\RankNext_T$ and a succinct representation of $\PLCP_T$.

\begin{lemma}[{\cite[Sect.\ 2.2.3]{fischer18lz}}]\label{lemLCPSST}
    Given a constant~$\epsilon$ with $0 < \epsilon \le 1$,
    there is a data structure that can access $\ISA_T[i]$ and $\LCP_T[i]$
    in $\pi_a(n) = O(\epsilon^{-1})$ time using $\pi_s(n) = (1+\epsilon)n \log n + O(n)$ bits of working space.
    It can be constructed within this working space in $\pi_c(n) = O(\epsilon^{-2} n)$ time.
\end{lemma}

\begin{lemma}\label{lemLCPCST}
    Given a constant~$\epsilon$ with $0 < \epsilon \le 1$,
    there is a data structure that can access $\ISA_T[i]$ and $\LCP_T[i]$
    in $\pi_a(n) = O(\log_\sigma^\epsilon n)$ time using $\pi_s(n) = O(\epsilon^{-1} n \log \sigma)$ bits of working space.
    It can be constructed within this working space in $\pi_c(n) = O(n \log_\sigma^\epsilon n)$ time.
\end{lemma}
\begin{proof}
    We use the compressed suffix tree presented in \cite{munro17cst}.
    While the succinct suffix tree naturally supports the required access queries,
    the compressed suffix tree needs to be enhanced with a sampling data structure to support access to $\ISA_T[i]$ and $\LCP_T[i]$.
    In detail, it provides access to the PLCP array $\PLCP_T$.
Our idea is to provide access to $\SA_T$ and $\ISA_T$ in $O(\log^\epsilon n)$ time.
    For the former, we use a sampling of $\SA_T$ with $O(\epsilon^{-1} n \log \sigma)$ bits to obtain access to $\SA_T$ with $O(\log^\epsilon_\sigma n)$ time~\cite[Sect.\ 3.2]{grossi05csa}.
    For the latter, having this $\SA_T$ representation, we create a representation of $\ISA_T$ using $O(\epsilon^{-1} n \log \sigma)$ space in $O(n \log^\epsilon_\sigma n)$ time~\cite{munro12inverse}.
    This representation can access $\ISA_T[i]$ in $O(\log^\epsilon n)$ time.
\end{proof}

\begin{lemma}\label{lemLCPBWT}
    There is a data structure that can access $\ISA_T[i]$ and $\LCP_T[i]$
    in $\pi_a(n) = O(\log r \log^{O(1)} n)$ time using $\pi_s(n) = r \log^{O(1)} n + O(n)$ bits of working space.
    It can be constructed within this working space in $\pi_c(n) = O(n + r \log^{O(1)} n)$ time.
\end{lemma}
\begin{proof}
Like in the proof of Lemma~\ref{lemLCPCST}, we provide access to $\LCP_T$ by $\PLCP_T$ and $\SA_T$:
First, we compute the run-length encoded BWT in $O(n \log r)$ time stored in $O(r \log n)$ bits of space~\cite{ohno17rlbwt}.
Next, we build a data structure answering
$\SA_T[i]$ and $\ISA_T[i]$ in $O(\log^{O(1)} n \log r)$ time
while using $O(r \log n + n)$ bits of space~\cite[Thm 5.1]{kempa19index}.
It can be constructed in $O(n + r \log^{O(1)} n)$ time with $O(n + r \log^{O(1)} n)$ bits of working space.

\end{proof}

\begin{lemma}\label{lemPhi}
  We can compute $\MB_T$ and $\ME_T$ in $O(n)$ time
  using $n\log n + O(\sigma \log n) + 4n + o(n)$ bits of total working space
  including the space for $\MB_T$ and $\ME_T$.
\end{lemma}
\begin{proof}
First, we construct $\RankPrev_T$ from $T$ using additional $O(\sigma \lg n)$ bits of working space~\cite{Goto2014LZ}. Our second target is the succinct representation of the PLCP array $\succPLCP_T$ of Sadakane~\cite{Sadakane2002Phi} using $2n + o(n)$ bits of space while allowing constant time random access to each PLCP value.
We can compute it from $\RankPrev_T$~\cite{Karkkainen2009PLCP}.
  After computing $\succPLCP_T$, we delete $\RankPrev_T$.
Subsequently, we construct $\RankNext_T$ from $T$ similarly to Goto et al.~\cite{Goto2014LZ}\footnote{They proposed a construction algorithm only for $\RankPrev_T$.
    However, we can modify the algorithm to construct $\RankNext_T$.
    Their algorithm simulates a suffix array construction algorithm, but
    uses $O(\sigma \lg n)$ pointers to positions in $\RankPrev_T$ at which new values are inserted.
    By these pointers, they represent $\RankPrev_T$ as $2\sigma$ linked lists.
    In \cite[Step 4 in Sect.~4.2.2]{Goto2014LZ}, they show how to invert these lists.
    Consequently, we can use this step to invert all lists after all elements have been inserted.
  }.
  Finally, we obtain $\succPLCP_T$ and $\RankNext_T$.
  All the steps can be executed in $O(n)$ time.

  In the following, we describe an algorithm for computing all MUSs of $T$ by using $\PLCP_T$ (in the succinct representation~$\succPLCP_T$) and $\RankNext_T$.
  As described in the proof of Lemma~\ref{lemBuildBitVectors},
  if we can compute $\ell_i = \max\{\LCP_T[\ISA_T[i]],$ $\LCP_T[\ISA_T[i]+1]\}$ in $\pi'_a(n)$ time
  for each text position $i$ with $1 \le i \le n$, then
  we can compute $\MB_T$ and $\ME_T$ in $O(n\cdot \pi'_a(n))$ time.
  Actually, we can achieve that $\pi'_a(n) = O(1)$ by using the fact that 
$\ell_i = \max\{\PLCP_T[i], \PLCP_T[\RankNext_T[i]]\}$.
Therefore, we can compute $\MB_T$ and $\ME_T$ in $O(n)$ time.
  Also, the total working space for the above procedure is $n\log n + O(\sigma \lg n) + 4n + o(n)$ bits
  including the space for $\MB_T$ and $\ME_T$.
\end{proof}
 \section{Conclusions}\label{sec:conc}
In this paper, we proposed compact data structures for
the interval SUS problem and the point SUS problem.
Our data structure is the first data structure of size $O(n)$ bits
for the interval SUS problem (resp. the point SUS problem).
On the one hand, for the interval SUS problem,
we proposed a data structure of size $2n+2m+o(n)$ bits
answering a query in output-sensitive $O(\occ)$ time,
where $n$ is the length of the input string,
$m$ is the number of MUSs in the input string, and
$\occ$ is the number of returned SUSs.
On the other hand, for the point SUS problem,
we proposed a data structure of size $\lceil(\log_2{3}+1)n\rceil+o(n)$ bits
answering a query in the same output-sensitive time.
The construction time and the working space of each data structures
depends on the complexity of simulating
the inverse suffix array and the LCP array of the input string.
For example, by using the succinct suffix tree~\cite{fischer18lz},
we can achieve $O(n)$ time and
$2n\log n + O(n)$ bits of working space
to construct our data structures.

\subsection{Open Problems}
We are unaware of an algorithm computing $\RankPrev_T$ or $\RankNext_T$ from the text~$T$ in-place for integer alphabets.
Since it is possible to store $\RankPrev_T$ in unary with the same representation as the succinct PLCP representation~\cite{Sadakane2002Phi} in $O(r)$ space~\cite{koeppl19lexparse}, 
where $r$ is the number of runs in the BWT, we wonder whether we can compute $\RankPrev_T$ in compressed space.
A possibility seems to adapt the in-place suffix array construction algorithm of Li et al.~\cite{Li2018SA}.

As future work, we want to extend our algorithm to compute SUSs with $k$ edits (or $k$ mismatches).
A SUS with $k$ edits is a substring that is unique even when changing $k$ arbitrary characters (allowing deletions, insertions and character exchanges).
Similarly, a SUS with $k$ mismatches is unique even when exchanging $k$ arbitrary characters.
It is known that SUSs with $k$ edits (resp.\ mismatches) contain SUSs with $k-1$ edits (resp.\ mismatches), where SUSs with no edits (resp.\ no mismatches) are the ordinary SUSs.

\subsection*{Acknowledgments}
This work was supported by JSPS KAKENHI Grant Numbers JP18F18120 (KD), JP18K18002 (YN), JP17H01697 (SI), JP16H02783 (HB), JP18H04098 (MT), and by JST PRESTO Grant Number JPMJPR1922 (SI).

\end{document}